\documentclass[conference]{IEEEtran}
\usepackage[numbers]{natbib}
\usepackage{multicol}
\usepackage[bookmarks=true]{hyperref}

\usepackage{subcaption}
\usepackage{tikz}
\usepackage{mathtools}
\usepackage{amsfonts}
\usepackage{amssymb}
\usepackage{nicefrac}
\usepackage[utf8]{inputenc}
\usepackage[english]{babel}
\usepackage{color}
\usepackage{standalone}
\usepackage[ruled,vlined]{algorithm2e}
\usetikzlibrary{shapes,backgrounds}
\usetikzlibrary{intersections,calc}
\usetikzlibrary{patterns}
\usepackage{pgfplots}
\pgfplotsset{compat=1.16}

\usepackage{pgfplotstable}

\newcommand{\XobsRandom}{\boldsymbol{X}_{\mathrm{obs}}}
\newcommand{\Xobsinit}{X_\mathrm{obs}[0]}
\newcommand{\ProbObs}{\mathbb{P}_{\XobsRandom}^{t,\Xobsinit}}
\newcommand{\SafeSet}{\mathcal{S}}
\newcommand{\algName}{\texttt{Safely}}
\newcommand{\xsqp}{{\{x_\mathrm{sqp}[t]\}}_{t=0}^T}
\newcommand{\xpr}{{\{x_\mathrm{pr}[t]\}}_{t=0}^T}
\newcommand{\oRange}{\mathbb{N}_{[1,N_{o}]}}

\usepackage{amsthm}
\newtheorem{assumption}{Assumption}
\newtheorem{prop}{Proposition}
\newtheorem{prob}{Problem}

\newtheorem{lem}{Lemma}

\title{
\texttt{Safely}: Safe Stochastic Motion Planning Under Constrained Sensing via Duality
}

\author{Michael Hibbard$^{1}$, Abraham P. Vinod$^{2}$, Jesse Quattrociocchi$^{1}$, and Ufuk Topcu$^{1}$

\thanks{$^{1}$M. Hibbard, J. Quattrociocchi, and U. Topcu are with the Department of Aerospace and Engineering Mechanics at the University of Texas at Austin.}%
\thanks{$^{2}$A.P. Vinod is with Mitsubishi Electric Research Laboratories.}%
}

\IEEEoverridecommandlockouts
\begin{document}

\maketitle
\thispagestyle{empty}
\pagestyle{empty}

\begin{abstract}
Consider a robot operating in an uncertain environment with stochastic, dynamic obstacles. 
Despite the clear benefits for trajectory optimization, it is often hard to keep track of each obstacle at every time step due to sensing and hardware limitations.
We introduce the $\algName{}$ motion planner, a receding-horizon control framework, that simultaneously synthesizes both a trajectory for the robot to follow as well as a sensor selection strategy that prescribes trajectory-relevant obstacles to measure at each time step while respecting the sensing constraints of the robot.
We perform the motion planning using sequential quadratic programming, and prescribe obstacles to sense based on the duality information associated with the convex subproblems.
We guarantee safety by ensuring that the probability of the robot colliding with any of the obstacles is below a prescribed threshold at every time step of the planned robot trajectory.
We demonstrate the efficacy of the $\algName{}$ motion planner through software and hardware experiments.

\end{abstract}

\section{INTRODUCTION}
\begin{figure}[t]
    \centering
    \begin{subfigure}[t]{0.999\columnwidth}
        \centering
        \scalebox{0.50}{
        \includegraphics{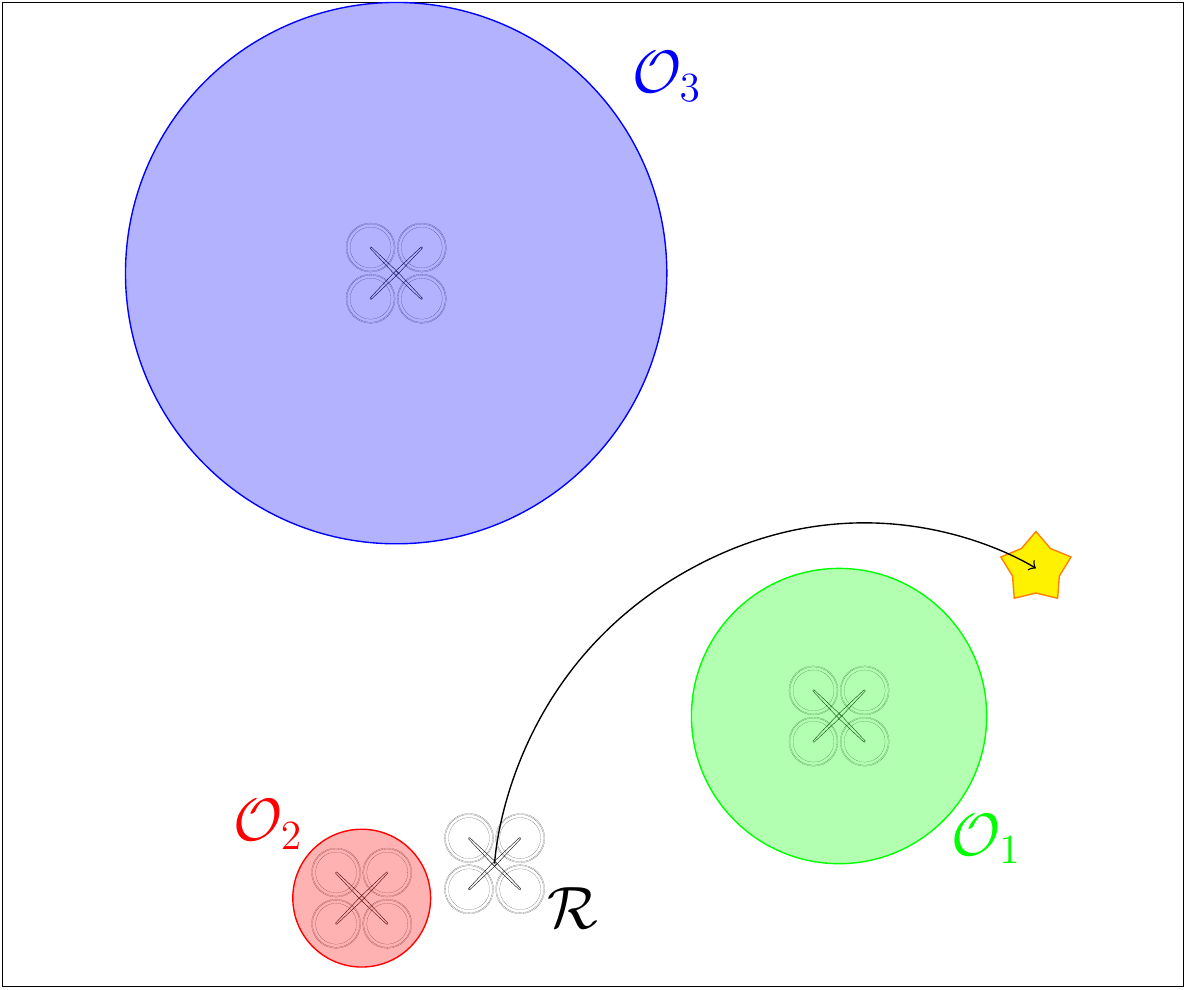}}
        \caption{Only obstacle $\mathcal{O}_1$ influences the optimal trajectory of the robot.}
        \label{fig:tradeoff}
    \end{subfigure}
    \begin{subfigure}[t]{0.999\columnwidth}
        \centering\scalebox{0.88}{
        \includegraphics{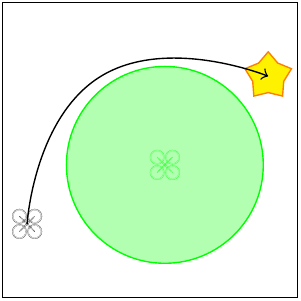}}
        \centering\scalebox{0.88}{
        \includegraphics{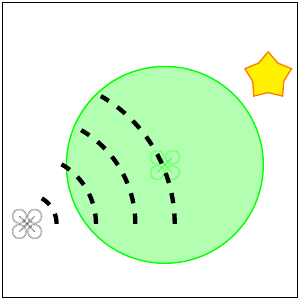}}
        \centering\scalebox{0.88}{
        \includegraphics{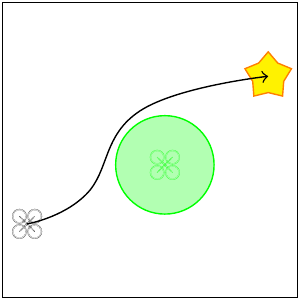}}
        \caption{Reducing the uncertainty associated with obstacle $\mathcal{O}_1$ enables a significant improvement in performance.}
        \label{fig:make_obs}
    \end{subfigure}
    \caption{The proposed algorithm, \algName{}, utilizes convex optimization and duality theory to simultaneously design both a safe motion plan that avoids the uncontrolled, dynamic obstacles in the environment, as well as a constrained sensor allocation that reduces the uncertainty in the underlying state of the obstacles that have the highest impact on performance.}
\end{figure}

We study the motion planning problem of a robot that operates in an environment populated by dynamic obstacles and has limited sensing capabilities. Despite these limitations, the robot seeks to remain safe by avoiding collisions with these obstacles.
Several relevant applications of such a problem include food delivery robots that must reach their destination while avoiding pedestrians and
self-driving cars that must generate real-time motion plans while remaining cognizant of other road vehicles.
For the purpose of motion planning, a popular method to predict the motion of the dynamic, uncontrolled obstacles present in the environment utilizes stochastic, dynamical models. 
Such models may be obtained from a combination of data-driven and model-driven approaches~\cite{bar2004estimation,rasmussen2003gaussian,aoude2013probabilistically}.
Without real-time measurements, the future obstacle state predictions obtained from these models typically have high variance, requiring overly conservative robot motion plans in order to satisfy the safety requirements.
However, reducing uncertainty through measurements of the underlying state of each obstacle at every time step is often impractical due to sensing and hardware limitations.
In the case of a food delivery robot, maintaining awareness of all pedestrians at all times can drive up hardware costs, perception-related computational costs, and planning time for the robot, resulting in poor reaction time and performance losses. 
To strike a balance between these extremes, we propose a \emph{motion planning algorithm that simultaneously designs a safe trajectory for the robot to complete its task and allocates its limited available sensors to reduce the uncertainty of the obstacles most relevant to the trajectory}.

We measure the relevance of an obstacle in terms of how the performance of the robot is affected by the uncertainty in the obstacle localization.
In Figure~\ref{fig:tradeoff}, the task of the drone $\mathcal{R}$, which we can control, is to reach its target state at the yellow star as quickly as possible.
The shaded circles represent areas in the environment where an obstacle resides with high probability, as determined by the stochastic model of the obstacle and past measurements.
By making a measurement about the underlying state of obstacle $\mathcal{O}_1$, and thereby reducing the uncertainty in its position, the robot can further reduce the time required to reach the target state, as shown in Figure~\ref{fig:make_obs}. 
Clearly, measuring the underlying states of obstacles $\mathcal{O}_2$ and $\mathcal{O}_3$ will have no effect on the performance of the drone, and the limited sensing resources available should not be spent on reducing the uncertainty associated with these obstacles.

Our main contribution is a novel stochastic motion planner, which we name $\algName{}$. Using $\algName{}$, we design a safe trajectory for the robot in the stochastic environment and identify the obstacles for which sensing resources should be allocated based on their influence on the performance of the currently planned trajectory. 
We ensure probabilistic safety in motion planning via probabilistic occupancy~\cite{vinod2018probabilistic,vinod2018stochastic}, which provides outer-approximations to the keep-out sets in the form of ellipsoids.
We formulate the stochastic motion planning problem as a nonlinear program, which we then solve using sequential quadratic programming~\cite{boggs1995sequential}.
We exploit the connection between duality theory for convex optimization and sensitivity analysis to design a constrained sensor allocation strategy that identifies the most relevant obstacles while respecting the sensing limitations.
\algName{} uses a Kalman filter \cite{bar2004estimation} to update the obstacle state estimates after measurement, wherein we assume linear sensor measurements and that the obstacles follow linear dynamics subject to Gaussian disturbances.

\textbf{Related work.} 
Perhaps the research most closely related to this paper is that of \cite{ghosh2017joint}, in which an agent checks only for the obstacles that intersect its planned trajectory. 
Similarly, \cite{Axelrod-RSS-17} provides probabilistic safety guarantees for uncertain obstacles.
However, these works consider only static obstacles and do not incorporate future uncertainty into the sensor selection process of the robot, whereas \algName{} does not have such restrictions.
The work of \cite{patil2014gaussian} likewise considers the relation between sensing and acting, but only considers the relevance of an obstacle in terms of the uncertainty in its state position and ignores its impact on the future planned trajectory.
Furthermore, \cite{tanaka2017lqg} studies the problem of minimizing the directed information from the plant output to the control input. However, this work does not consider this problem in the context of obstacle avoidance as we do.

More generally, our work is reminiscent of active perception \cite{bajcsy1988active}, wherein an agent behaves in such a way as to gain information by controlling what it senses. However, the existing approaches for perception-related tasks often do not account for safety constraints.
Moreover, the proposed approach of utilizing the dual variables of obstacle avoidance constraints to prescribe the sensor selection strategy of the robot is novel, to the best of our knowledge. 
Recently, researchers have looked into deep learning to design perceptual control policies that account for attention~\cite{pmlr-v100-lee20b}, but enforce safety only by handing over the control to an expert or stopping before the obstacle. Such safety strategies are not feasible for autonomous robots that operate independently of an expert.
Existing literature has also proposed receding-horizon control frameworks for the stochastic motion planning problem under the assumption of full-state feedback of the obstacles~\cite{vinod2018stochastic,du2011robot}.

The remainder of the paper is structured as follows. In Section~\ref{sec:prelim_and_prob}, we begin by formulating the stochastic motion planning problem of the robot and how it can utilize the measurement of an obstacle state to improve its performance. In Section~\ref{sec:motion_planner}, we propose the \algName{} motion planner using sequential quadratic programming and probabilistic occupancy, and motivate the use of the dual variables in the sensor selection strategy using sensitivity analysis. Finally, in Section \ref{sec:num_ex}, we provide software and hardware experiments demonstrating the benefit of incorporating the \algName{} motion planner for obstacle avoidance.

\section{PRELIMINARIES and PROBLEM FORMULATION}\label{sec:prelim_and_prob}

We denote the set of real numbers by $\mathbb{R}$ and the set of all positive real numbers by $\mathbb{R}_{+}$.
Additionally, we denote the set of natural numbers by $\mathbb{N}$, and we use $\mathbb{N}_{[a,b]}$ to denote a sequence of natural numbers $\{a,a+1,\ldots,b\}$.
Finally, we denote the set of all positive semidefinite matrices of size $n \times n$ by $\mathbb{S}_{+}^{n}$.
Furthermore, let $\mathbf{0}_{m\times n}$ be an $m$$\times$$n$ matrix of zeros and let $\mathbf{I}_{n}$ be the $n$-dimensional identity matrix.
For two sets $\mathcal{A}$ and $\mathcal{B}$, we use the operator $\oplus$ to denote the Minkowski sum $\mathcal{A}$$\oplus$$\mathcal{B}$$=$$\{a$$+$$b|a$$\in$$\mathcal{A}$$,$$b$$\in$$\mathcal{B}\}$. Throughout the paper, we will use $\|\cdot\|$ to denote the Euclidean norm, and we denote the standard Euclidean ball by $\mathrm{Ball}(\mu,r)\triangleq\{x\in \mathbb{R}^{n}: ||x-\mu||\leq r\}$.

\subsection{Robot and environment model}

Consider a robot operating in an environment where it has the discrete-time deterministic dynamics%
\begin{equation}\label{eq:agent_dynamics}
    x[t+1] = f(x[t],u[t]),
\end{equation}
with state $x[t]\in \mathbb{R}^{n}$, input $u[t]\in\mathcal{U} \subset \mathbb{R}^{m}$, and known nonlinear, differentiable function $f:\mathbb{R}^{n+m}\rightarrow \mathbb{R}^{n}$. The goal of the robot is to reach $\epsilon$-close to a specified target state $x_{g}\in \mathbb{R}^n$. Additionally, the robot must not collide with the $N_O\in\mathbb{N}$ dynamic obstacles present in the environment. For any $o\in\mathbb{N}_{[1,N_O]}$, the $o^\mathrm{th}$ obstacle has the following known, discrete-time, stochastic, uncontrolled, linear dynamics, 
\begin{equation}\label{eq:obstacle_dynamics}
    x_{o}[t+1] = A_{o}x_{o}[t] + B_{o}w_{o}[t],
\end{equation}
with state $x_o[t]\in\mathbb{R}^n$ and disturbance $w_{o}[t]\sim\mathcal{N}(\mu_{w_{o}},\Sigma_{w_{o}})$.
Here, the matrices $A_{o}\in\mathbb{R}^{n\times n}$ and $B_{o}\in\mathbb{R}^{n\times p}$, respectively, while the vector $\mu_{w_o}\in \mathbb{R}^{p}$ and $\Sigma_{w_o}\in\mathbb{S}^{p}_{+}$ encode our known prior belief about the motion of the $o^\mathrm{th}$ obstacle.
Furthermore, we assume the robot has access to a stochastic prior belief over the initial state of the $o^\mathrm{th}$ obstacle, $x_o[0]\sim\mathcal{N}(\mu_{o}[0],\Sigma_{o}[0])$, where $\mu_{o}\in\mathbb{R}^{n}$ and $\Sigma_{o}[0]\in\mathbb{S}_{+}^{n}$.
Such a model for the obstacle dynamics is common in the Kalman filter literature \cite{bar2004estimation}.

%
We assume that if the robot allocates sensor resources to observe obstacle $o$ at time $t$, then the robot makes an observation according to the linear sensor model
\begin{align}\label{eq:linear_sensor_model}
    z_{o}[t+1] = H x_{o}[t] + \nu[k],
\end{align}
where $z_{o}[t]\in\mathbb{R}^{q}$, $H\in\mathbb{R}^{q \times n}$, and $\nu\sim\mathcal{N}(\mu_{\nu},\Sigma_{\nu})$ is the measurement noise, with $\mu_{\nu} \in \mathbb{R}^{q}$ and $\Sigma_{\nu} \in \mathbb{S}_{+}^{q}$.
In this paper, we assume $\mu_{\nu}=\mathbf{0}_{q\times1}$.
Given this observation, the robot updates its estimate $\hat{x}_{o}[t+1]$ of $x_{o}[t+1]$ as well as the covariance $\Sigma_{o}[t+1]$ according to the standard Kalman filter update \cite{bar2004estimation},
\begin{subequations}
\begin{align}
    x_{o}^{-}[t+1] & = A_{o} \hat{x}_{o}[t] + B_{o} \mu_{w_{o}}, \label{eq:xtp1_minus} \\
    \Sigma_{o}^{-}[t+1] & = A_{o} \Sigma_{o}[t] A_{o}^{\top} + B_{o} \Sigma_{w_{o}} B_{o}^{\top}, \label{eq:sigmatp1_minus} \\
    K_{o}[t] & = \Sigma_{o}^{-}[t+1] H^{\top} \left( H \Sigma_{o}^{-}[t+1] H^{\top} + \Sigma_{\nu} \right)^{-1}, \label{eq:kf_gain} \\
    \hat{x}_{o}[t+1] & = x_{o}^{-}[t+1] - K_{o}[t] (z_{o}[t+1] - H \hat{x}_{o}[t]), \label{eq:xtp1} \\
    \Sigma_{o}[t+1] & = (\mathbf{I}_{n} - K_{o}[t]H) \Sigma_{o}^{-}[t+1]. \label{eq:sigmatp1}
\end{align}
\end{subequations}
Furthermore, in the absence of observations, we can recursively express $x_{o}[t]$, the obstacle state at a future time instant $t\in\mathbb{N}$, as
\begin{subequations}
\begin{align}
    x_{o}[t] &\sim\mathcal{N}(x_{o}^{-}[t],\Sigma^{-}_{o}[t]),  \label{eq:gauss_x_t}\\
    x_{o}^{-}[t] & = A_{o} x_{o}^{-}[t-1] + B_{o} \mu_{w_o} , \label{eq:gauss_mean_t} \\
    \Sigma_{o}^{-}[t] & = A_{o} \Sigma_{o}^{-}[t-1] A_{o}^{\top} + B_{o} \Sigma_{w_{o}} B_{o}^{\top}, \label{eq:gauss_sig_t}
\end{align}\label{eq:gauss_obs}%
\end{subequations}
%
which again follows from Kalman filter theory.
For ease of notation, we let $\mu_{o}[t] \triangleq x_{o}^{-}[t]$ and we denote the configuration of the obstacle set at a time step $t $$ \in $$ \mathbb{N}$ using the concatenated obstacle state vector $\XobsRandom[t] $$ \triangleq$$ \left[x_{1}^{\top}[t],\ldots,x_{N_O}^{\top}[t]\right]^{\top} \in \mathbb{R}^{(N_O n)}$ for which we define its associated probability measure as $\ProbObs$, where $\Xobsinit$ is the known initial obstacle state vector, and the obstacle state and uncertainty propagate according to~\eqref{eq:gauss_obs}.

\subsection{Optimal control problem with probabilistic safety}

We now discuss the optimal control problem of the robot. We start with the following simplifying assumption:
\begin{assumption}\label{assum:translate}
The robot and each of the $N_{o}$ obstacles are restricted to translational motion, and their rigid bodies do not rotate.
\end{assumption}
Assumption~\ref{assum:translate} is common in robotics applications, and can be relaxed to include rotational motion by considering symmetric covers for the rigid bodies~\cite[Sec 4.3.2]{lavalle2006planning}\cite{vinod2018probabilistic,du2011robot}.

Let the sets $\mathcal{R}_\mathrm{robot}$ and $\mathcal{R}_o$ for every $o\in\mathbb{N}_{[1,N_O]}$, denote the rigid bodies of the robot and the obstacles, respectively.
Under Assumption~\ref{assum:translate}, for a given robot trajectory $\{x[t]\}_{t=0}^T \triangleq \{ x[0],x[1],\ldots,x[T] \}$ of length $T\in\mathbb{N}$, the probability of collision with any of the $N_O$ obstacles over the course of the entire trajectory of the robot is given by the function $\mathrm{CollidePr}:\mathbb{R}^{(nT)}\to [0,1]$, where
\begin{align}
    &\mathrm{CollidePr}(x[0],\ldots,x[T];\Xobsinit) \nonumber\\
    &\quad \triangleq\ProbObs\{\cup_{o=1}^{N_O}\cup_{t=1}^{T} \{x_o[t] \in x[t] \oplus \mathcal{B}_{o}\}\}, \label{eq:collision_pr}
\end{align}
with sets $\mathcal{B}_o\triangleq \mathcal{R}_\mathrm{robot} \oplus (-\mathcal{R}_o) \subset \mathbb{R}^n$ for all $o\in\mathbb{N}_{[1,N_O]}$\cite{vinod2018probabilistic}. 

For some performance metric $J:\mathbb{R}^{((n+m)T)}\to \mathbb{R}$ and a user-defined safe operation zone $\SafeSet\subseteq \mathbb{R}^n$, we aim to solve the following optimal control problem of the robot,
\begin{subequations}
    \begin{align}
        \min & \quad J(x[1],\ldots,x[T],u[0],\ldots,u[T-1]) \label{eq:opt_cost_func} \\
        \mathrm{s.t.} &\quad u[\cdot] \in \mathcal{U}, x[\cdot] \in \SafeSet, \label{eq:opt_prob_basic_constraints} \\
        &\quad \text{Dynamics \eqref{eq:agent_dynamics} and \eqref{eq:obstacle_dynamics}}, \label{eq:opt_dynamics}\\
        &\quad\mathrm{CollidePr}(x[0],\ldots,x[T];\Xobsinit) \leq \alpha, \label{eq:opt_collide_prob}
    \end{align}\label{prob:orig_opt}%
\end{subequations}%
with decision variables $x[1],\ldots,x[T]$ and $u[0],\ldots,u[T-1]$. Here, $\alpha$$\in$$[0,1]$ is a (small) user-specified upper bound on the acceptable collision probability.
The cost function $J$ is assumed to be twice continuously differentiable and can represent a variety of performance metrics, such as quickly reaching the target state.
%
\begin{assumption}\label{assum:ball}
For every $o\in\mathbb{N}_{[1,N_O]}$, $\mathcal{B}_o=\mathrm{Ball}(c_o, r_o)$, for some center $c_o \in \mathbb{R}^n$ and some radius $r_o>0$.
\end{assumption}
Assumption~\ref{assum:ball} is not overly restrictive, since the rigid bodies of the robot and the obstacles typically have finite volume, which result in  bounded $\mathcal{B}_o$. 
Furthermore, we consider a covering hypersphere $\mathcal{B}_0\supseteq \mathcal{R}_\mathrm{robot} \oplus (-\mathcal{R}_o)$ when $\mathcal{R}_\mathrm{robot} \oplus (-\mathcal{R}_o)$ is not a hypersphere to satisfy Assumption~\ref{assum:ball}, which tightens the probabilistic safety constraint \eqref{eq:opt_collide_prob}.

\subsection{Problem formulation}

We now state the problem of interest for this paper:
\begin{prob}\label{prob_st:prob}
    Consider the sensing constraint that the robot can measure the underlying state of only $K<N_O$ obstacles at any time step. Design a motion planning algorithm that determines a safe trajectory for the robot with stochastic, dynamic obstacles, 
    while using the limited sensing resources to reduce the conservativeness in the future motion plans of the robot. 
\end{prob}

The sensor allocation constraint in Problem~\ref{prob_st:prob} is motivated by the sensing limitations of mobile robots with on-board computational hardware used for processing perception data.
The robot must balance the increasing risk of collision with obstacles due to the increasing uncertainty in the obstacle states with the performance of the future motion plans.  
Any motion planner that addresses Problem~\ref{prob_st:prob} is naturally amenable to a receding-horizon framework that further improves the tractability of the motion planning problem.

\section{The \algName{} motion planner}\label{sec:motion_planner}

In this section, we introduce the \algName{} motion planner to address Problem~\ref{prob_st:prob}. We first formulate the safe, observation-free motion planning problem of the robot.
We then analyze the utility of observing a particular obstacle through the lens of duality and sensitivity analysis, and propose a constrained sensor scheduling strategy. 
We conclude with a brief discussion of the numerical implementation of \algName{}.


\subsection{Nonlinear stochastic motion planning}
\label{sub:dc_mp}

In this section, we begin by discussing ellipsoidal outer-approximations of the keep-out sets, whose avoidance guarantees the desired probabilistic safety in \eqref{eq:opt_collide_prob}. These outer-approximations follow from probabilistic occupancy functions~\cite{vinod2018probabilistic}, and generalize sufficient conditions discussed in~\cite{vinod2018stochastic} for a conservative enforcement of joint chance constraints  \eqref{eq:opt_collide_prob}. We formalize these connections in Proposition~\ref{prop:1}.

\begin{prop}\label{prop:1}
    A sufficient condition to satisfy \eqref{eq:opt_collide_prob} is the following collection of constraints
    %
   \begin{align}
        (x[t]-\mu_{o}[t])^{\top} (Q_{o}^{+}[t])^{-1} (x[t]-\mu_{o}[t]) \geq 1 \label{eq:reverse_cvx_coll}
   \end{align}
    for each obstacle $o\in\mathbb{N}_{[1, N_O]}$ and each time step $t \in \mathbb{N}_{[1,T]}$. Here, $Q_{o}^{+}[t]\in\mathbb{R}^{n\times n}$ are positive definite matrices defined by the stochastic dynamics of the corresponding obstacle and a user-specified unit vector $l_{0}\in\mathbb{R}^n$,
    \begin{subequations}
    \begin{align}
        Q_{o}^{+}[t] &= \left( \sqrt{l_{0}^{\top}Q_{o}[t]l_{0}} + r_{o} \right) \left( \frac{Q_{o}[t]}{\sqrt{l_{0}^{\top}Q_{o}[t]l_{0}}} + r_o I_{n} \right), \label{eq:Q_matrix_plus} \\
        Q_{o}[t] &= -2\log\left( \frac{\alpha \sqrt{|2\pi\Sigma_{o}[t]|}}{TN_O\mathrm{Volume}(\mathrm{Ball}(c_o,r_{o}))}\right) \Sigma_{o}[t]. \label{eq:Q_matrix}
    \end{align}\label{eq:Q_defns}%
    \end{subequations}
\end{prop}
\begin{proof}
    Using Boole's inequality, we have
    \begin{align}
        &\ProbObs\{\cup_{o=1}^{N_O}\cup_{t=1}^{T} \{x_o[t] \in x[t] \oplus \mathcal{B}_{o}\}\} \nonumber\\
        &\quad\leq \sum\nolimits_{t=1}^T \sum\nolimits_{o=1}^{N_O}\ProbObs\{x_o[t] \in x[t] \oplus \mathcal{B}_{o}\}
    \end{align}
    for any robot trajectory ${\{x[t]\}}_{t=0}^T$. Clearly, a sufficient condition to guarantee \eqref{eq:opt_collide_prob} is that the robot trajectory must satisfy 
    \begin{align}
     \ProbObs\{x_o[t] \in x[t] \oplus \mathcal{B}_{o}\} &\leq \frac{\alpha}{T N_O},\label{eq:opt_collide_prob_icc}
    \end{align}
    for every $o\in\mathbb{N}_{[1,N_O]}$ and $t\in\mathbb{N}_{[1,T]}$. We conclude the proof by noting that \eqref{eq:opt_collide_prob_icc} is satisfied whenever $x[t]$ satisfies \eqref{eq:reverse_cvx_coll} for every $o\in\mathbb{N}_{[1,N_O]}$ and $t\in\mathbb{N}_{[1,T]}$.
    Using similar arguments as those used in~\cite[Section 3.A]{vinod2018stochastic} and ellipsoidal computational geometry, the avoidance of the ellipsoid \eqref{eq:reverse_cvx_coll} is sufficient to guarantee \eqref{eq:opt_collide_prob} under Assumption~\ref{assum:ball}.
\end{proof}
Using Proposition~\ref{prop:1}, we obtain the following tightened reformulation of \eqref{prob:orig_opt},
%
    \begin{align}
    \begin{array}{rl}
        \min &\quad \eqref{eq:opt_cost_func} \\
        %
        \mathrm{s.t.} &\quad
        \eqref{eq:opt_prob_basic_constraints}, \eqref{eq:opt_dynamics}, \\
        %
        \substack{\forall t \in \mathbb{N}_{[1,T]}, \\ \forall o\in\mathbb{N}_{[1,N_O]}}
        &\quad \eqref{eq:reverse_cvx_coll},
    \end{array}\label{prob:orig_opt_ell}%
        \end{align}
%
with decision variables $x[1],\ldots,x[T]$ and $u[0],\ldots,u[T-1]$. 
Specifically, every feasible solution of \eqref{prob:orig_opt_ell} is feasible for \eqref{prob:orig_opt}. 
%

We can also express \eqref{prob:orig_opt_ell} in the general form of a nonlinear program
\begin{align}
    \begin{array}{rl}
        \min & \quad c(y) \\
        \mathrm{s.t.} & \quad h(y) = \mathbf{0}_{j\times1},  \\
        & \quad g(y) \leq \mathbf{0}_{l\times 1},
    \end{array}\label{prob:gen_nonlinear}%
\end{align}
where $y\in \mathbb{R}^{i}$, $c:\mathbb{R}^{i}$$\rightarrow$$\mathbb{R}$, $h:\mathbb{R}^{i}$$\rightarrow$$\mathbb{R}^{j}$, and $g:\mathbb{R}^{i}$$\rightarrow$$\mathbb{R}^{l}$, and $c$, $h$, and $g$ are all  twice continuously differentiable.

A variety of methods exist to solve the nonlinear program~\eqref{prob:gen_nonlinear} to local optimality.
In this paper, we focus our attention on sequential quadratic programming (SQP)-based approaches \cite{boggs1995sequential}, wherein a sequence of quadratic programming (QP) subproblems are iteratively solved until convergence to a local optima is obtained.
At each iterate, these subproblems are constructed using a quadratic approximation of the objective function subject to the constraints linearized about the previous solution iterate.
Let $y^{k}$ denote the optimal solution of the QP at the $k^{th}$ iteration.
Then, $y^{k+1}$ is obtained by solving the following QP,
\begin{align}
\begin{array}{rl}
    \min & \quad \nabla c(y^{k})^{\top}(y^{k+1}-y^{k}) \\
    & \qquad + (y^{k+1} - y^{k})^{\top} M_{k} (y^{k+1} - y^{k}) \\
    \mathrm{s.t.} & \quad h(y^{k}) + \nabla h(y^{k})^{\top} (y^{k+1} - y^{k}) = \mathbf{0}_{j\times1}, \\
    & \quad g(y^{k}) + \nabla g({y^{k}})^{\top} (y^{k+1} - y^{k}) \leq \mathbf{0}_{l\times 1},
\end{array}\label{prob:gen_sqp_qp}
\end{align}
where $M^{k}$$\in$$\mathbb{R}^{i\times i}$ is a positive definite approximation to the Hessian of the Lagrangian $\nabla^{2}\mathcal{L}(y^{k},\lambda,\gamma)$, where
\begin{align*}
    \mathcal{L}(y^{k},\lambda,\gamma) = c(y^{k}) - \lambda^{\top} h(y^{k}) - \gamma^{\top} g(y^{k}),
\end{align*}
in which $\lambda$$\in$$\mathbb{R}^{j}$ and $\gamma$$\in$$\mathbb{R}^{l}$ are the Lagrange multipliers.
Note that since $M^{k}$ is positive definite and all constraints are affine in $y^{k+1}$,~\eqref{prob:gen_sqp_qp} is a convex QP and can be solved efficiently using off-the-shelf solvers such as \texttt{OSQP}~\cite{stellato2020osqp}.

Define the concatenated vectors $\mathbf{x}$$=$$[x[1]^{\top}, \ldots, x[T]^{\top}]^{\top}$ and $\mathbf{u}$$=$$[u[0]^{\top}, \ldots, u[T-1]^{\top}]^{\top}$. 
Then, \algName{} solves the following
QP at every iteration of the SQP procedure:
%
\begin{subequations}
\begin{align}
    \min & \quad \nabla J(\mathbf{x}^{k},\mathbf{u}^{k})^{\top} \left[ \begin{array}{cc}
         \mathbf{x}^{k+1} - \mathbf{x}^{k}  \\
         \mathbf{u}^{k+1} - \mathbf{u}^{k}
    \end{array}\right] \nonumber \\
    & \qquad + \left[ \begin{array}{cc}
         \mathbf{x}^{k+1} - \mathbf{x}^{k}  \\
         \mathbf{u}^{k+1} - \mathbf{u}^{k}
    \end{array}\right]^{\top} M^{k} \left[ \begin{array}{cc}
         \mathbf{x}^{k+1} - \mathbf{x}^{k}  \\
         \mathbf{u}^{k+1} - \mathbf{u}^{k}
    \end{array}\right] \label{eq:opt_cost_func_upd} \\
    \mathrm{s.t.} & { } \nonumber \\
    \substack{\forall t \in \mathbb{N}_{[1,T]}, \\ \forall o\in\mathbb{N}_{[1,N_O]}} & \quad  (C^{k}_{obs}[t])^{\top} x^{k+1}[t] \leq d^{k}_{obs}[t], \label{eq:opt_dc_lin} \\
    \forall t\in \mathbb{N}_{[0,T-1]} & \quad x^{k+1}[t+1] = \\ 
    & \qquad (C^{k}_{dyn}[t])^{\top} \left[ \begin{array}{cc}
         x^{k+1}[t]  \\
         u^{k+1}[t]
    \end{array} \right] + d^{k}_{dyn}, \label{eq:opt_agent_dyn} \\
    \substack{\forall t \in \mathbb{N}_{[0,T-1]}, \\ \forall o\in\mathbb{N}_{[1,N_O]}} & \quad \mu_{o}[t+1] = A_{o} \mu_o[t] + B_{o} \mu_{w_o}, \label{eq:opt_gauss_mean_t} \\
    \substack{\forall t \in \mathbb{N}_{[0,T-1]}, \\ \forall o\in\mathbb{N}_{[1,N_O]}} & \quad \Sigma_{o}[t+1] = A_{o} \Sigma_o[t] A_{o}^\top + B_{o} \Sigma_{w} B_{o}^T, \label{eq:opt_gauss_sig_t} \\
    \forall t\in \mathbb{N}_{[0,T]} & \quad u^{k+1}[t] \in \mathcal{U}, x^{k+1}[t] \in \SafeSet, \label{eq:opt_safe_upd}
\end{align}\label{prob:opt_diff_convex}%
\end{subequations}%
where 
\begin{align*}
    C^{k}_{obs}[t] & = 2(Q_{o}^{+})^{-1} (x^{k}[t] - \mu_{o}[t]), \\
    d^{k}_{obs}[t] & = 1 - (x^{k}[t] - \mu_{o}[t])^{\top} (Q_{o}^{+})^{-1} (x^{k}[t] - \mu_{o}[t]) \\
    & \qquad + 2 (Q_{o}^{+})^{-1} (x^{k}[t] - \mu_{o}[t]) x^{k}[t], \\
    C^{k}_{dyn}[t] & = \nabla f(x^{k}[t],u^{k}[t]), \\
    d^{k}_{dyn}[t] & = f(x^{k}[t],u^{k}[t]) - (\nabla f(x^{k}[t],u^{k}[t]))^{\top} \left[ \begin{array}{cc}
         x^{k}[t]  \\
         u^{k}[t]
    \end{array} \right],
\end{align*}
in which~\eqref{eq:opt_dc_lin} and~\eqref{eq:opt_agent_dyn} are the linearizations of the ellipsoidal keep-out constraints \eqref{eq:reverse_cvx_coll} and the robot dynamics~\eqref{eq:agent_dynamics}, respectively, about the previous solution iterate of the SQP procedure. In what follows, we denote the solution of the SQP procedure by $\xsqp$. By construction, $\xsqp$ is a dynamically-feasible trajectory for the robot that also satisfies the probabilistic safety constraint \eqref{eq:opt_collide_prob}.

\subsection{Duality-based constrained sensor scheduling}\label{sec:dual_based_sensing}

We now turn our attention to the problem of scheduling sensors under constraints. 
Intuitively, the robot can dramatically improve its overall performance in future time steps by reducing the uncertainty in the states of the obstacles that adversely affect its current motion. 
By reducing the uncertainty of these obstacles, the robot is free from adopting  conservative trajectories for the sake of safety.
However, since the keep-out constraints \eqref{eq:reverse_cvx_coll} are enforced via linearized approximations, the impact of the measurement of an obstacle state on $\xsqp$ is unclear, as shown in Figure~\ref{fig:dc_bef_aft_obs}. We propose a refinement of $\xsqp$ using the \emph{project-and-linearize} approach~\cite{vinod2018stochastic}. We show that the influence of a measurement of the obstacle state on the refinement is easily established via sensitivity analysis.


\begin{figure}
    \centering\scalebox{0.77}{
    \includegraphics{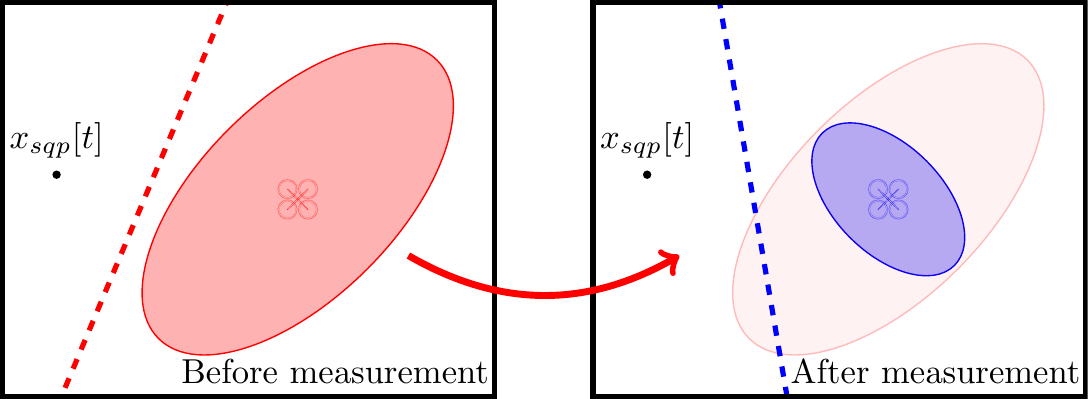}}
    \caption{Effect of measurement on the linearized obstacle-avoidance constraints.}
    \label{fig:dc_bef_aft_obs}
\end{figure}

We refine the solution $\xsqp$ by solving the following nonlinear program, 
\begin{subequations}
\begin{align}
    \min & \quad J(x[1],\ldots,x[T],u[0],\ldots,u[T-1]) \label{eq:opt_cost_func_proj} \\
    \mathrm{s.t.} & \quad C_\mathrm{proj}\left[x[0]^{\top},\ldots, x[T]^{\top} \right]^{\top} \leq d_\mathrm{proj}, \label{eq:opt_dc_lin_proj} \\
    &\quad \eqref{eq:opt_prob_basic_constraints}, \eqref{eq:opt_dynamics}, \eqref{eq:opt_gauss_mean_t}, \eqref{eq:opt_gauss_sig_t}. \nonumber
\end{align}\label{prob:opt_proj}%
\end{subequations}
Here, the matrix $C_\mathrm{proj}\in\mathbb{R}^{(TN_O)\times (Tn)}$ and the vector $d_\mathrm{proj}\in\mathbb{R}^{(TN_O)}$ define a convex feasible set defined by the supporting hyperplanes at the projection of ${\{x_\mathrm{sqp}[t]\}}_{t=0}^T$ onto the keep-out ellipsoids \eqref{eq:reverse_cvx_coll}. The optimization problem \eqref{prob:opt_proj} is guaranteed to find a solution, which we will denote as $\xpr$, no worse than $\xsqp$, since $\xsqp$ is feasible for \eqref{prob:opt_proj}. 

While the safety constraints in \eqref{prob:opt_diff_convex} are enforced via a  linearization of the non-convex quadratic constraint \eqref{eq:reverse_cvx_coll} at each iteration of the SQP procedure, the safety constraints in \eqref{prob:opt_proj} are enforced via the supporting hyperplanes of the keep-out ellipsoids \eqref{eq:reverse_cvx_coll}. Consequently, the effect of shrinking the keep-out sets directly influences $d_\mathrm{proj}$ in contrast to the linearized constraints. 

The effect of obstacle state measurement can also be viewed as a relaxation of the nonlinear program \eqref{prob:opt_proj}, wherein $d_\mathrm{proj}$ is replaced by $d_\mathrm{proj} + \delta$ for some slack $\delta\in\mathbb{R}^{(TN_O)}$, and $\delta \succeq 0$. 
%
After obtaining a new measurement, the maximum eigenvalues of the covariance matrix of the obstacle position shrinks due to the subtraction of the Kalman gain term in~\eqref{eq:sigmatp1}. 
This reduction in the covariance in turn shrinks the keep-out ellipsoids \eqref{eq:reverse_cvx_coll} of the obstacle, relaxing the obstacle avoidance constraints in the motion-planning problem of the robot.


We now utilize sensitivity analysis, along with a mild assumption, to design a constrained sensor scheduling strategy.
\begin{lem}[\textsc{Sensitivity analysis~\cite[\textrm{Sec. 5.6}]{boyd2004convex}}]\label{lem:duality}
Consider the $\delta$-perturbed, convex optimization problem,
\begin{subequations}
\begin{align}
    \min &\quad  f_0(x) \\
    \mathrm{s.t.} &\quad x\in \mathcal{X},\\
                         &\quad f_i(x) \leq \delta_i\ \forall i\in\mathbb{N}_{[1,N]},
\end{align}\label{prob:perturb}%
\end{subequations}
with decision variable $x\in\mathbb{R}^n$, convex $f_i \ \forall i\in\mathbb{N}_{[0, N]}$, convex set $\mathcal{X}$, and perturbation vector $\delta=[\delta_1\ \ldots\ \delta_N]\in \mathbb{R}^N$. Let $p^\ast(\delta):\mathbb{R}^N \to \mathbb{R}$ denote the optimal value of \eqref{prob:perturb}. Then, whenever strong duality holds for \eqref{prob:perturb} with $\delta=0$, we have 
\begin{align}
    p^\ast(\delta) \geq p^\ast(0) - {\left({\lambda^\ast}\right)}^\top \delta,
\end{align}
for every $\delta\in\mathbb{R}^N$, where $\lambda^\ast\in \mathbb{R}^N$ is the optimal dual variable vector  corresponding to \eqref{prob:perturb} with $\delta=0$. In addition to strong duality, if $p^\ast$ is differentiable at $\delta=0$, then relaxing the $i^\mathrm{th}$ constraint $f_i(x)\leq 0$ to $f_i(x)\leq \delta_i$, for some $\delta_i > 0$, approximately decreases the optimal value by $\lambda_i^\ast\delta_i$.
\end{lem}

\begin{assumption}\label{assum:Slaters}
Each QP subproblem of \eqref{prob:opt_proj} satisfies Slater's condition, i.e., it has a strictly feasible solution.
\end{assumption}
The validity of Assumption~\ref{assum:Slaters} can be easily checked via the positivity of the Chebyshev radius of the polytope corresponding to the feasible set of \eqref{prob:opt_proj}. Recall that the computation of the Chebyshev radius is a linear program~\cite[Ch. 8]{boyd2004convex}. In all of our simulation and hardware experiments, we found that Assumption~\ref{assum:Slaters} was always satisfied.
Under Assumption~\ref{assum:Slaters}, strong duality always holds for each QP subproblem of the SQP procedure for the refined nonlinear program in \eqref{prob:opt_proj}. Furthermore, the optimal solution map $p^\ast$ for \eqref{prob:opt_proj} is always differentiable via the KKT conditions~\cite[Ch. 5]{boyd2004convex}. 
%
%
%

\emph{Constrained sensor scheduling strategy:} 
By Lemma~\ref{lem:duality} and Assumption~\ref{assum:Slaters}, the obstacle measurement that corresponds to the greatest impact on the improvement of the performance of the robot should correlate to the obstacle with the largest values for its associated dual variables. Consequently, for each obstacle $o\in\mathbb{N}_{[1,N_O]}$, we first compute a scalar $\Lambda_o=\sum_{t=1}^T \gamma^t \lambda_{o,t}^\ast$ using the optimal dual variables. Here, $\gamma\in(0, 1]$ serves as a discount factor that allows the user to balance the present uncertainty in the state of an obstacle with its increasingly uncertain future state position. We subsequently measure the underlying states of the obstacles corresponding to the top $K$ values of $\Lambda_o$. Based on the measurements, we recompute their associated (smaller) keep-out ellipsoids \eqref{eq:reverse_cvx_coll}. For the obstacles that were not measured, we continue using the original keep-out ellipsoids for motion planning.

\subsection{Numerical implementation}

\begin{algorithm}[t]
    \caption{\algName{} motion planner}\label{algo:safely}
    \DontPrintSemicolon
    
    Input: robot dynamics~\eqref{eq:agent_dynamics}, initial state $x[0]$, target state $x_g$, tolerance $\epsilon$, feasible control input $\mathcal{U}$, and safe set $\mathcal{S}$; obstacle shape $\mathcal{B}_o$, dynamics~\eqref{eq:obstacle_dynamics}, and initial belief $x_{o}[0]\sim \mathcal{N}(\mu_{o}[0],\Sigma_{o}[0])$ for each $o\in \mathbb{N}_{[1,N_{o}]}$; maximum probability $\alpha$, planning horizon $T$, discount factor $\gamma$, observations $K$ 
    
    \While{$|x[t]-x_g| \geq \epsilon$}{
        
        \texttt{Propagate, construct ellipsoids:} \\
            \For{$o=1\ldots N_{o}$}{
                \For{t=1\ldots T}{
                    Compute $\mu_{o}[t]$,$\Sigma_{o}[t]$ by \eqref{eq:gauss_mean_t},\eqref{eq:gauss_sig_t} \\
                    Compute $Q_{o}[t]$ by \eqref{eq:Q_matrix} \\
                    Compute $Q_{o}^{+}[t]$ by \eqref{eq:Q_matrix_plus} \\
                    Store $\mathcal{E}(\mu_{o}[t],Q_{o}^{+}[t])$
                }
            }
        \texttt{Solve SQP:} \\
        \While{$|J_{new}-J_{old}| \geq \hat{\epsilon}$}{
            $J_{new}$ = $\mathrm{min}$ \eqref{prob:opt_diff_convex}
        }
        $\xsqp=\mathrm{argmin}$\eqref{prob:opt_diff_convex} \\
        \texttt{Projection:} \\
            Compute $C_{proj},d_{proj}$ \\
            $\xpr = \mathrm{argmin}$ \eqref{prob:opt_proj} \\
            Compute $\Lambda_{o}[t]$ for $o\in \mathbb{N}_{[1,N_{o}]}$ \\
        \texttt{Observe:} \\
            \For{$o=1\ldots N_{o}$}{
                \If{$\Lambda_{o}\in K \text{ largest}$ }{
                    Draw observation $z_{o}[1]$ according to \eqref{eq:linear_sensor_model} \\
                    Update $x_{o}[1]$, $\Sigma_{o}[1]$ according to \eqref{eq:xtp1_minus}-\eqref{eq:sigmatp1}
                }
            }
        \texttt{Reset:} \\
        $x[0]=x_{\mathrm{pr}}[1]$ \\
        \For{$o=1\ldots N_{o}$}{
            $\mu_{o}[0]=\mu_{o}[1]$, $\Sigma_{o}[0]=\Sigma_{o}[1]$ 
        }
    }
\end{algorithm}

%

\begin{figure}[t]
    \centering
    \begin{subfigure}[t]{0.45\textwidth}
        \centering
        \scalebox{0.75}{
        \includegraphics{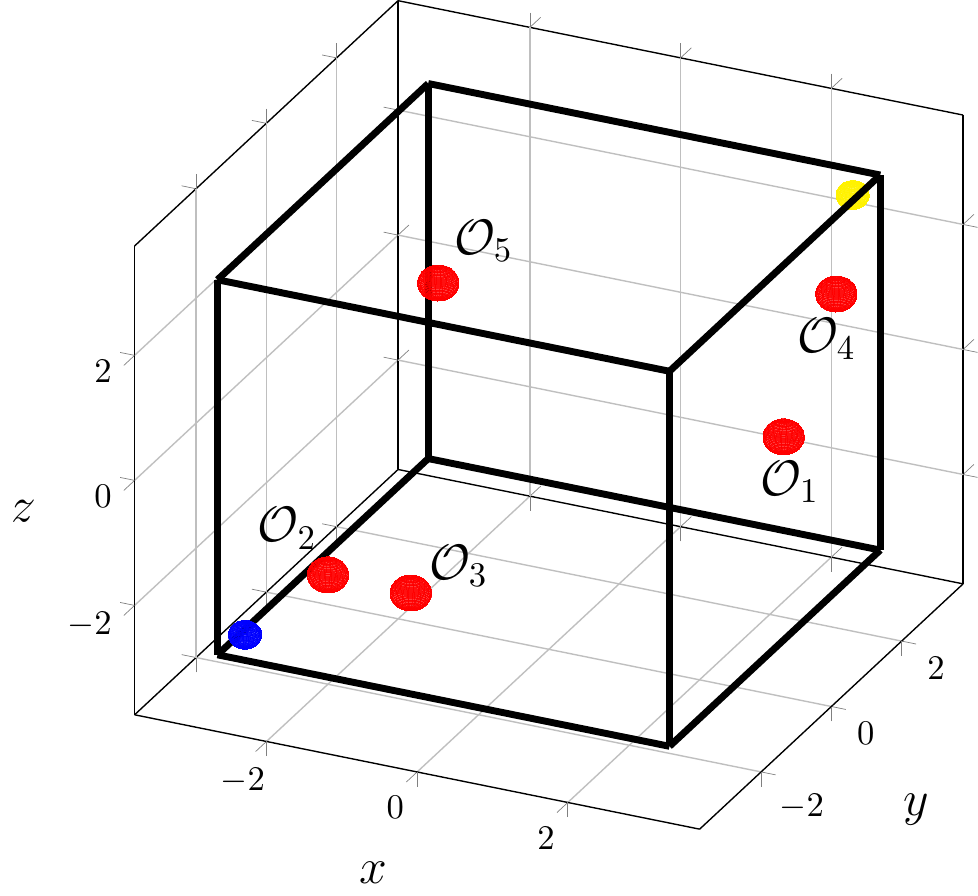}}
        \caption{Initial configuration of Example 1.}
        \label{fig:example_env}
    \end{subfigure}
    \begin{subfigure}[t]{0.45\textwidth}
    \centering
        \scalebox{1}{
        \includegraphics{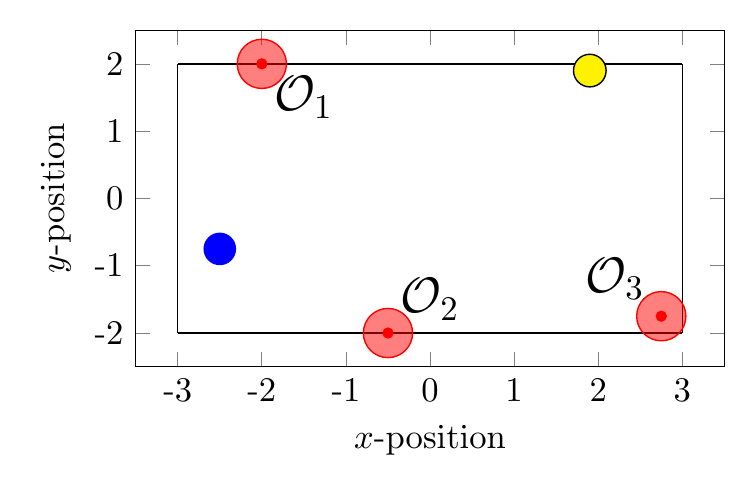}}
        \caption{Initial configuration of Example 2.}
        \label{fig:nonlin_example_env}
    \end{subfigure}
    \caption{Initial configurations of the sample environments.}
    \label{fig:simulation_init_configs}
\end{figure}

\begin{figure*}[!ht]
    \centering
    \begin{subfigure}[t]{0.32\textwidth}
        \centering        
        \scalebox{2.25}{\includegraphics{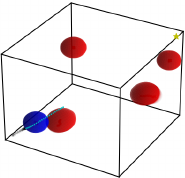}}
        \caption{$t=10$}
        \label{fig:sum_10}
    \end{subfigure}
    \begin{subfigure}[t]{0.32\textwidth}
    \centering
        \scalebox{2.25}{\includegraphics{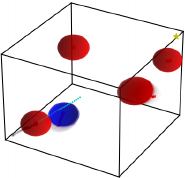}}
        \caption{$t=20$}
        \label{fig:sum_20}
    \end{subfigure}
    \begin{subfigure}[t]{0.32\textwidth}
        \centering
        \scalebox{2.25}{
        \includegraphics{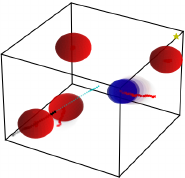}}
        \caption{$t=30$}
        \label{fig:sum_30}
    \end{subfigure}
    \begin{subfigure}[t]{0.32\textwidth}
        \centering
        \scalebox{2.25}{
        \includegraphics{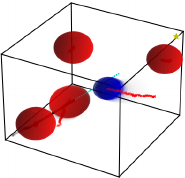}}
        \caption{$t=40$}
        \label{fig:sum_40}
    \end{subfigure}
    \begin{subfigure}[t]{0.32\textwidth}
        \centering
        \scalebox{2.25}{
        \includegraphics{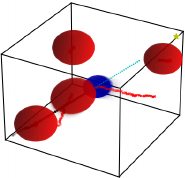}}
        \caption{$t=50$}
        \label{fig:sum_50}
    \end{subfigure}
    \begin{subfigure}[t]{0.32\textwidth}
        \centering
        \scalebox{2.25}{
        \includegraphics{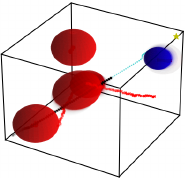}}
        \caption{$t=60$}
        \label{fig:sum_60}
    \end{subfigure}
    \caption{Trajectory obtained using the \algName{} motion planner at various time steps over the simulation. The black marks indicate the trajectory of the robot up to the current time step, while the light blue marks indicate the future motion plan. The red marks indicate the realization of the random obstacle trajectories up to the current time step. The ellipsoids indicate the keep-out regions for guaranteeing probabilistic safety (Proposition~\ref{prop:1}). The obstacle ellipsoid selected for observation at a particular time step is colored blue, and otherwise is colored red. }
    \label{fig:plot_traj}
\end{figure*}

We summarize the \algName{} motion planner used to solve \eqref{prob:orig_opt} in Algorithm~\ref{algo:safely}. 
Note that since we assume the obstacle dynamics are linear Gaussian systems, we have closed-form expressions for all terms in the propagation and the construction step~\cite{vinod2018stochastic}. 
Furthermore, for the ellipsoidal outer-approximations, we can compute the projections used to construct $C_{\mathrm{proj}}$ and $d_{\mathrm{proj}}$ by solving
\begin{align*}
    \min_{x\in \mathbb{R}^{n}} & \quad ||x - x_{sqp}[t]||^{2} \\
    \mathrm{s.t.} & \quad (x - \mu_{o}[t])^{\top} (Q_{o}^{+}[t])^{-1} (x - \mu_{o}[t]) \leq 1,
\end{align*}
for each $o\in \mathbb{N}_{[1,N_{o}]}$ and each $t \in \mathbb{N}_{[1,T]}$, which is a convex second-order cone program that can be solved efficiently using off-the-shelf solvers like \texttt{ECOS}~\cite{domahidi2013ecos}.

\section{EXPERIMENTAL VALIDATION}\label{sec:num_ex}

We now consider several numerical experiments demonstrating the \texttt{Safely} motion planner.

\subsection{Software experiment using linear robot dynamics}

We consider the environment shown in Figure~\ref{fig:example_env}, where the safe set $\SafeSet$ of the robot is the $6\text{m}\times6\text{m}\times6\text{m}$ black cube centered at the origin. The robot begins in the lower-left corner in state $x[0]=[-2.75,-2.75,-2.75]^{\top}$ and must travel to the target state located at $x_{g}=[2.75,2.75,2.75]^{\top}$. We use the quadratic cost function
\begin{equation}\label{eq:quad_cost}
    J(x,u) = \sum\nolimits_{t=0}^{T}||x[t]-x_{g}||^{2}
\end{equation}
to model this objective, where $T$ is the planning horizon. We consider a fixed planning horizon of $T=25$ time steps.

While travelling towards its target state, the robot must avoid colliding with any of a set of five stochastic, dynamic obstacles whose initial positions are provided in the Appendix. We assume that the robot begins with perfect knowledge of these values; i.e., $\Sigma_{o}[0]=\mathbf{0}_{3\times3}$, and that the radius of each ball $\mathcal{B}_{o}$ is given by $r_{o}=0.25$m for each $o \in \oRange$. This initial obstacle configuration yields the five red ellipsoids shown in Figure~\ref{fig:example_env}.

The robot follows the linear dynamics
\begin{align*}
    x[t+1] = A x[t] + B u[t],
\end{align*}
where $x[t] \in \mathbb{R}^{6}$, $u[t] \in \mathbb{R}^{3}$, and
%
\begin{align*}
    A = \begin{bmatrix}
        \mathbf{I}_{3} & dt \cdot \mathbf{I}_{3} \\
        \mathbf{0}_{3\times3} & \mathbf{I}_{3}
    \end{bmatrix}, 
    B = \begin{bmatrix}
        \nicefrac{dt^{2}}{2} \cdot \mathbf{I}_{3} \\
        dt \cdot \mathbf{I}_{3}
    \end{bmatrix},
\end{align*}
in which $dt$ is the time interval between discrete steps, here set to $dt=0.25$s. Note that $x[t]$ is decomposed as $x[t]$$=$$[p[t]^{\top},v[t]^{\top}]^{\top}$, where $p[t],v[t]$$\in$$ \mathbb{R}^{3}$ are the position and the velocity components of the state vector, respectively.
The parameters for the sensor model~\ref{eq:linear_sensor_model} are given by
\begin{align}\label{eq:sim_sens_model}
    H = \mathbf{I}_{3}, \quad \mu_{\nu} = \left[ 0, 0, 0 \right]^{\top}, \quad \Sigma_{\nu} = 0.05 \mathbf{I}_{3}.
\end{align}
Furthermore, each obstacle $o \in \oRange$ follows the linear dynamics~\eqref{eq:obstacle_dynamics} with
%
\begin{align}\label{eq:obs_dynamics_example}
    A_{o} = \mathbf{I}_{3}, \quad
    B_{o} = dt \cdot \mathbf{I}_{3}.
\end{align}
The control input $u[t]=[u_{1}[t],u_{2}[t],u_{3}[t]]^{\top}$ of the robot is constrained to $\mathcal{U} = \{ u_{1}[t], u_{2}[t], u_{3}[t] \, : ||u[t]||_{\infty} \leq 0.5 \}$. Finally, the parameters of the Gaussian disturbances for each obstacle are provided in the Appendix.
%


We allow a single underlying obstacle state to be measured at each time step (setting $K=1$) and use a discount factor of $\gamma=1$. Furthermore, we set the value of $\alpha$ in~\eqref{eq:opt_collide_prob} to $0.01$. We use the \texttt{CASADI} Python interface  \cite{andersson2019casadi} with the \texttt{WORHP} SQP solver \cite{buskens2012esa} to solve \eqref{prob:opt_diff_convex}. Furthermore, since the state dynamics are linear and the cost function is quadratic, the refined problem in \eqref{prob:opt_proj} reduces to a QP, which we solve using \texttt{OSQP}. The simulation is run on a 1.8 GHz Intel Core i7-8550U CPU with 16 GB RAM. 
Figure~\ref{fig:plot_traj} displays the resulting trajectory obtained using the \algName{} motion planner. The robot required a total of $102$ time steps to reach the target state from its initial position, requiring an average computation time of $0.139$ seconds per iteration of \algName, with Figure~\ref{fig:times_stem_plot} showing the computation time required at each iteration. 

\begin{figure}
    \centering
    \scalebox{0.8}{
    \includegraphics[]{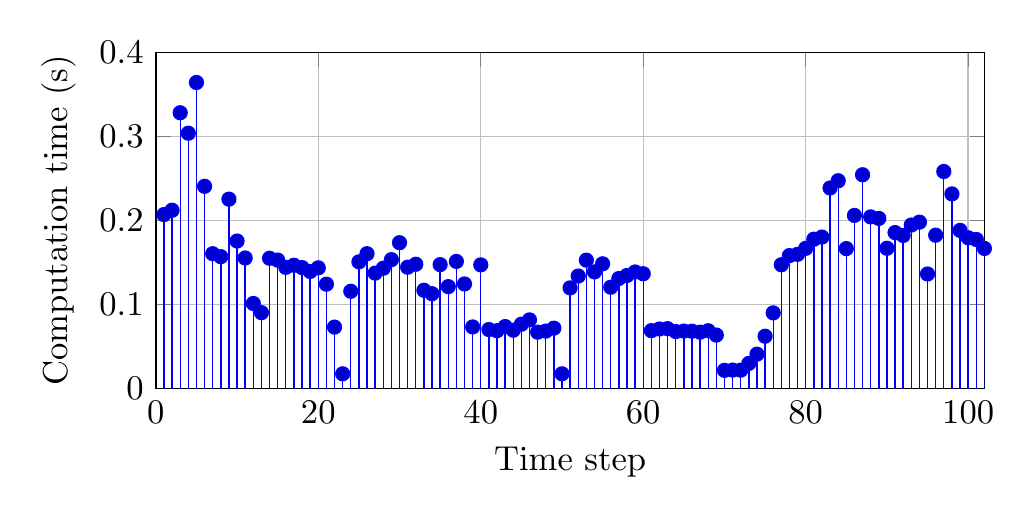}}
    \caption{Computation times for the $\algName{}$ motion planner at each time step of the simulation.}
    \label{fig:times_stem_plot}
\end{figure}
\begin{figure}
    \centering
    \scalebox{0.88}{
    \includegraphics[]{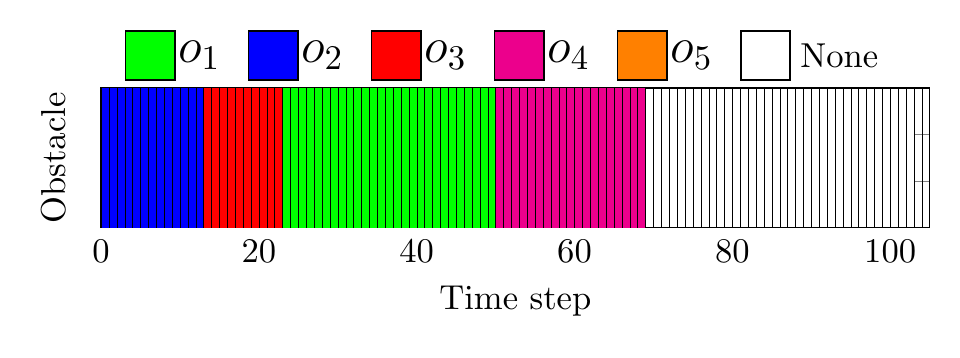}}
    \caption{Obstacle identified as the most relevant by the $\algName{}$ motion planner at each time step of the simulation.}
    \label{fig:obstacle_obs_made_plot}
\end{figure}

From Figures~\ref{fig:plot_traj} and~\ref{fig:obstacle_obs_made_plot}, the robot initially allocates its sensing resources towards observing $\mathcal{O}_2$ as it immediately impacts its nominal trajectory. Once the robot has navigated around $\mathcal{O}_{2}$, this obstacle no longer affects its trajectory and obstacle $\mathcal{O}_{3}$ is subsequently identified as the most relevant. However, this obstacle is only briefly the most relevant, as obstacle $\mathcal{O}_{1}$ quickly begins to significantly affect the final portion of the nominal trajectory. Once past $\mathcal{O}_{1}$, $\mathcal{O}_{4}$ is identified as the most relevant obstacle and is subsequently observed. Once the robot navigates around obstacle $\mathcal{O}_{4}$, no further sensor resources are required, as the robot is able to travel unimpeded to its target state. Note that obstacle $\mathcal{O}_{5}$ never affects the nominal trajectory of the robot and hence never has sensor resources allocated towards its observation.


\subsection{Software experiment using nonlinear robot dynamics}

\begin{figure}[t]
    \centering
    \scalebox{1.25}{
    \includegraphics{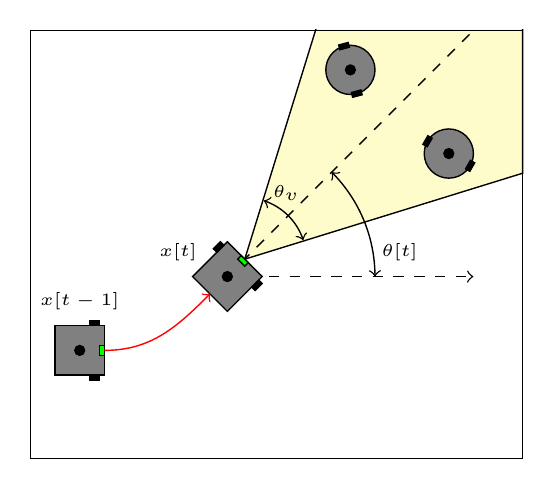}}
    \caption{Visualization of robot motion and sensing for nonlinear dynamics considered with a fixed, forward-facing camera.}
    \label{fig:fixed_camera_motion}
\end{figure}

\begin{figure*}[!ht]
    \centering
    \begin{subfigure}[t]{0.24\textwidth}
        \centering        
        \scalebox{1.4}{\includegraphics{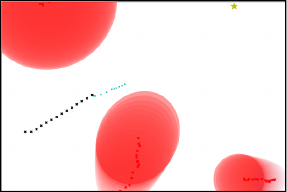}}
        \caption{$t=15$, $\algName$}
        \label{fig:safely_15}
    \end{subfigure}
    \begin{subfigure}[t]{0.24\textwidth}
    \centering
        \scalebox{1.4}{\includegraphics{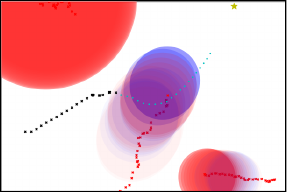}}
        \caption{$t=30$, $\algName$}
        \label{fig:safely_30}
    \end{subfigure}
    \begin{subfigure}[t]{0.24\textwidth}
        \centering
        \scalebox{1.4}{
        \includegraphics{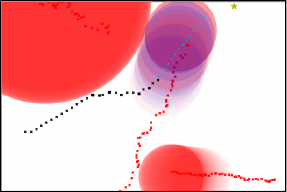}}
        \caption{$t=45$, $\algName$}
        \label{fig:safely_45}
    \end{subfigure}
    \begin{subfigure}[t]{0.24\textwidth}
        \centering
        \scalebox{1.4}{
        \includegraphics{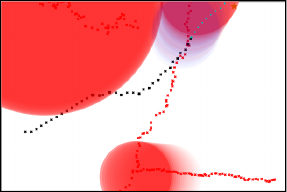}}
        \caption{$t=60$, $\algName$}
        \label{fig:safely_60}
    \end{subfigure}
    \begin{subfigure}[t]{0.24\textwidth}
        \centering
        \scalebox{1.4}{
        \includegraphics{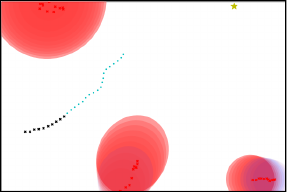}}
        \caption{$t=9$, BSP}
        \label{fig:bsp_9}
    \end{subfigure}
    \begin{subfigure}[t]{0.24\textwidth}
        \centering
        \scalebox{1.4}{
        \includegraphics{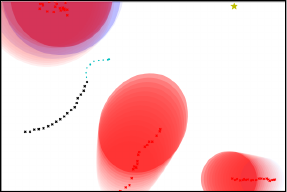}}
        \caption{$t=18$, BSP}
        \label{fig:bsp_18}
    \end{subfigure}
    \begin{subfigure}[t]{0.24\textwidth}
        \centering
        \scalebox{1.4}{
        \includegraphics{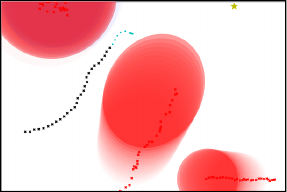}}
        \caption{$t=27$, BSP}
        \label{fig:bsp_27}
    \end{subfigure}
    \begin{subfigure}[t]{0.24\textwidth}
        \centering
        \scalebox{1.4}{
        \includegraphics{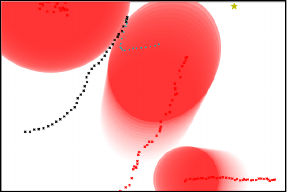}}
        \caption{$t=36$, BSP}
        \label{fig:bsp_36}
    \end{subfigure}
    \caption{Comparison of trajectories based on method considered --- \algName{} (a)--(d) vs. belief-space planning (BSP) (e)--(h). Since \algName{} only seeks to observe an obstacle when it is deemed relevant, it is able to reach the goal state since it does not attempt to observe irrelevant obstacles.}
    \label{fig:plot_traj_sim_2}
\end{figure*}

\begin{figure}
    \centering\scalebox{0.8}{
    \includegraphics{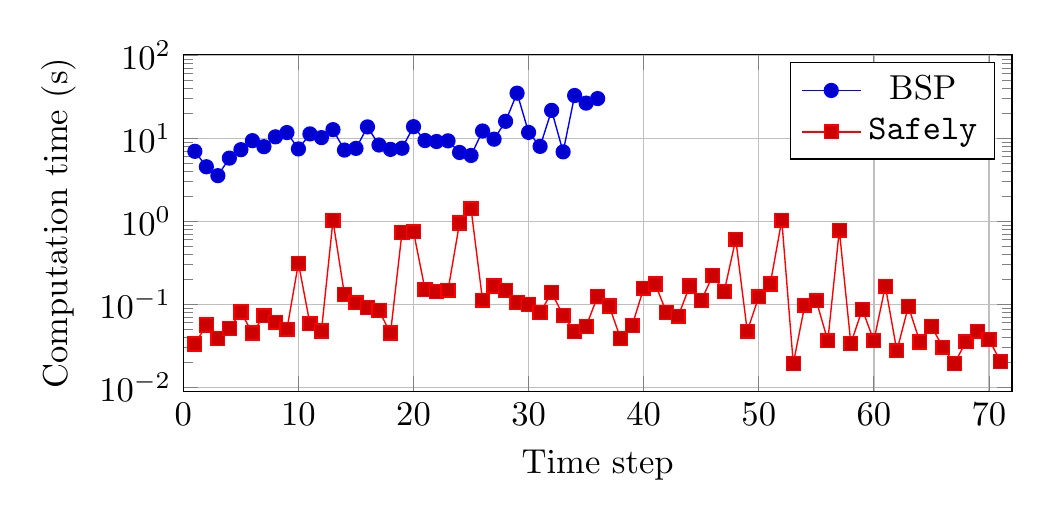}}
    \caption{Computation times for each iteration of \algName{} and the belief-space planning (BSP) method for Example 2.}
    \label{fig:comparison_solve_times}
\end{figure}

We now apply the \texttt{Safely} motion planner to the case that the robot follows the dynamics of Dubins' vehicle \cite{shapiro2014lectures}, where
\begin{subequations}
\begin{align}
    x[t+1] & = x[t] + dt \begin{bmatrix} v[t] \cos{\theta[t]} \\ v[t] \sin{\theta[t]} \end{bmatrix} \label{eq:dubins_1}, \\
    \theta[t+1] & = \theta[t] + dt \cdot \omega[t], \label{eq:dubins_2}
\end{align}\label{eq:dubins_dyn}%
\end{subequations}%
in which $v[t]\in [v_{min},v_{max}] \subset \mathbb{R}$ is the velocity input in m/s, $\theta[t]\in(-\pi,\pi]$ is the heading angle, and $\omega[t]\in[\omega_{min},\omega_{max}] \subset \mathbb{R}$ is the turning rate in rad/s. We consider these dynamics since they allow us to examine situations where the orientation of the robot affects what it can sense. Specifically, we consider the case that the robot uses a fixed, forward-facing camera with viewing angle $\theta_{v}$, as shown in Figure~\ref{fig:fixed_camera_motion}. 
When \texttt{Safely} identifies a relevant obstacle to sense, the robot must adjust its motion plan to locate this obstacle within its field of view (FOV). 
Since there is an explicit coupling between sensing and planning, we adjust the objective function in~\eqref{eq:quad_cost} to
\begin{equation}\label{eq:nonlin_cost_func}
    \sum\nolimits_{t=1}^{T} \beta \hat{\gamma}^{t} \langle \mu_{o_{r}}[t] - x[t], u_{\theta}[t] \rangle + ||x[t]-x_{g}||^{2},
\end{equation}
where $u_{\theta}[t] \triangleq [\cos(\theta[t]), \sin(\theta[t])]^{\top}$, $\beta \in \mathbb{R}_{+}$ is a weighting parameter, $\hat{\gamma} \in (0,1]$ is a discount factor, and $\mu_{o_{r}}[t]$ refers to the mean position at time $t$ of the current most relevant obstacle. The first term in~\eqref{eq:nonlin_cost_func} is the inner product between the heading angle of the robot and the direction vector from the robot to the most relevant obstacle's mean position, which incentivizes the robot to turn toward the most relevant obstacle. Note that if no obstacle is identified as relevant at a given time step, then the first term in~\eqref{eq:nonlin_cost_func} is omitted for that iteration of \algName. By discounting the first term, the resulting obstacle dual variables should continue to indicate the relevance of an obstacle in regard to the primary task of reaching the target state.

In this example, we consider the environment shown in Figure~\ref{fig:nonlin_example_env}, where the safe set $\SafeSet$ of the robot is the $6$m$\times$$4$m black square centered at the origin. The robot begins in the lower-left in state $x[0]=[-2.75,-1.00]^{\top}$ and must travel to the target state located at $x_{g}=[1.9,1.9]^{\top}$. We use $dt=0.5$s and $T=20$ time steps, and set $\beta=10$ and $\hat{\gamma}=0.8$. Finally, we set $v[t]\in [0.01,0.25]$m/s, $\omega[t]\in[-\frac{\pi}{3},\frac{\pi}{3}]$rad/s, and set $\theta_{v}=\frac{\pi}{3}$. Furthermore, the obstacles follow the linear dynamics
\begin{align}\label{eq:obs_dynamics_example}
    A_{o} = \mathbf{I}_{2}, \quad
    B_{o} = dt \cdot \mathbf{I}_{2}.
\end{align}
and their parameters are again provided in the Appendix. Finally, the parameters for the sensor model~\ref{eq:linear_sensor_model} are given by
\begin{align}\label{eq:sim_sens_model_ex_2}
    H = \mathbf{I}_{2}, \quad \mu_{\nu} = \left[ 0, 0 \right]^{\top}, \quad \Sigma_{\nu} = 0.05 \mathbf{I}_{2}.
\end{align}
All remaining parameters are the same as the previous example. Since the objective and the constraints are nonconvex, we use \texttt{WORHP} to solve both \eqref{prob:opt_diff_convex} and \eqref{prob:opt_proj}. 

As a comparison, we additionally implement a modified version of the belief-space planning (BSP) method proposed in~\cite{patil2014gaussian}, which considers obstacle uncertainty in the motion planning problem but not their relevance. Specifically, we modify the BSP framework by removing the hard end-point constraint and imposing obstacle avoidance by enforcing the ellipsoidal keep-out constraints corresponding to the $Q^{+}$ matrices constructed using the previous iterate's output from line 9 of Algorithm 1 in~\cite{patil2014gaussian}. We again use \texttt{WORHP} to solve the SQP subproblems in each iteration of the BSP algorithm.
Figure~\ref{fig:plot_traj_sim_2} displays the resulting trajectories over the course of each simulation. For \algName{}, the robot only allocates sensor resources towards sensing obstacle $\mathcal{O}_{2}$ since it travels directly through the shortest path between the initial state and the target state of the robot. After initially turning towards this obstacle at time step $15$, as shown in Figure~\ref{fig:safely_15}, the robot soon observes it, reducing the uncertainty in its state position. The robot must turn towards this obstacle several more times as it continues to plan to travel above obstacle $\mathcal{O}_{2}$. After sufficiently many observations of $\mathcal{O}_{2}$, the robot eventually plans to travel beneath it instead, as shown in Figure~\ref{fig:safely_45}. Figure~\ref{fig:safely_45} additionally shows that the robot continues to observe $\mathcal{O}_{2}$, as otherwise its uncertainty grows to cover the target state. By time step $60$, the robot has an unimpeded trajectory to the target state, as shown in Figure~\ref{fig:safely_60}.

Without the notion of obstacle relevance, the robot using the BSP method is unable to reach the target state without violating the collision-avoidance constraint. Initially, in Figure~\ref{fig:bsp_18}, the robot aligns itself with obstacle $\mathcal{O}_{1}$ as doing so provides the greatest uncertainty reduction while still adequately moving towards the target state. Doing so, however, quickly causes the robot to become ``pinched" between obstacles $\mathcal{O}_{1}$ and $\mathcal{O}_{2}$. By time step $36$, the robot has no safe nominal trajectory and subsequently violates the collision-avoidance constraint, as shown in Figure~\ref{fig:bsp_36}. Through this example, it is evident that obstacle relevance cannot simply be captured by the uncertainty in the position of that obstacle.

We also note that since the belief-space planning method requires solving multiples iterations of a more complex SQP at each time step, it requires significantly more computation time. Although our implementation of the BSP method is not optimized for computational efficiency, Figure~\ref{fig:comparison_solve_times} highlights the significantly lower computation times per iteration required by \algName{}. Whereas \algName{} required $0.182$ seconds per iteration, the BSP method required $11.81$ seconds per iteration.

\subsection{Hardware-based experiment using velocity inputs and turtlebots}

\begin{figure}
    \centering
    \scalebox{0.55}{\includegraphics{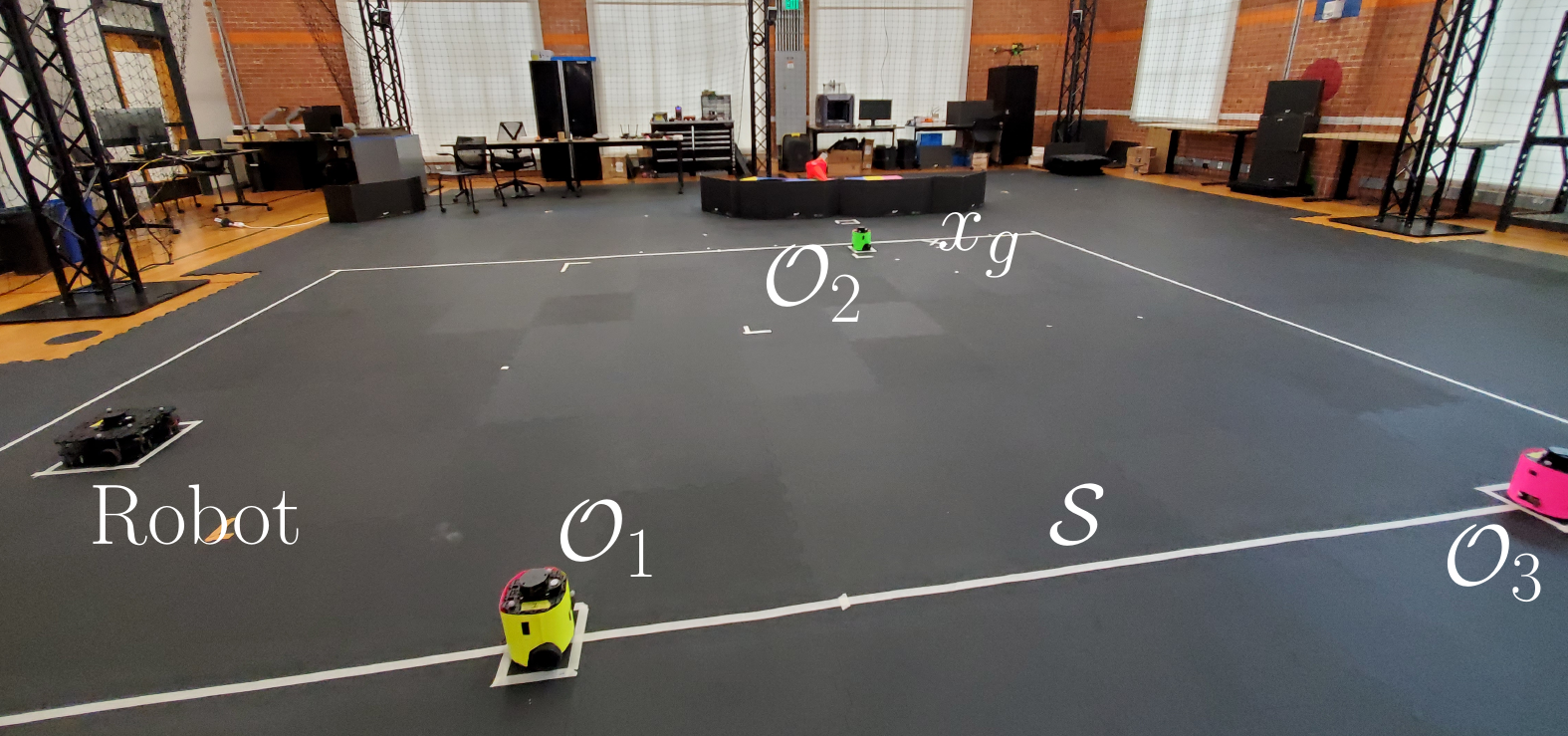}}
    \caption{Initial configuration of the hardware experiment.}
    \label{fig:hardware_init}
\end{figure}

\begin{figure*}[!ht]
    \centering
    \begin{subfigure}[t]{0.24\textwidth}
        \centering        
        \scalebox{1.4}{\includegraphics{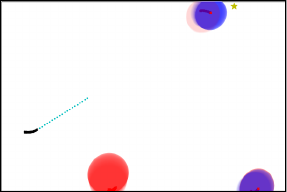}}
        \caption{$t=25$}
        \label{fig:hardware_25}
    \end{subfigure}
    \begin{subfigure}[t]{0.24\textwidth}
    \centering
        \scalebox{1.4}{\includegraphics{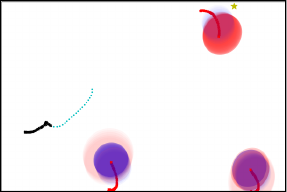}}
        \caption{$t=50$}
        \label{fig:hardware_50}
    \end{subfigure}
    \begin{subfigure}[t]{0.24\textwidth}
        \centering
        \scalebox{1.4}{
        \includegraphics{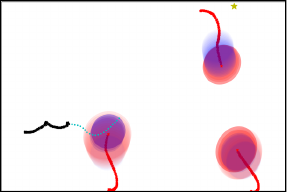}}
        \caption{$t=75$}
        \label{fig:hardware_75}
    \end{subfigure}
    \begin{subfigure}[t]{0.24\textwidth}
        \centering
        \scalebox{1.4}{
        \includegraphics{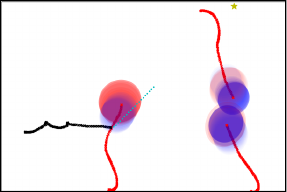}}
        \caption{$t=100$}
        \label{fig:hardware_100}
    \end{subfigure}
    \begin{subfigure}[t]{0.24\textwidth}
        \centering
        \scalebox{1.4}{
        \includegraphics{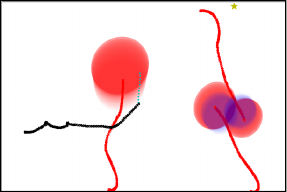}}
        \caption{$t=125$}
        \label{fig:hardware_125}
    \end{subfigure}
    \begin{subfigure}[t]{0.24\textwidth}
        \centering
        \scalebox{1.4}{
        \includegraphics{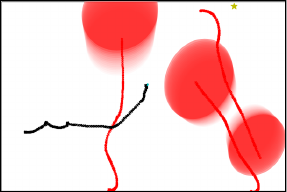}}
        \caption{$t=150$}
        \label{fig:hardware_150}
    \end{subfigure}
    \begin{subfigure}[t]{0.24\textwidth}
        \centering
        \scalebox{1.4}{
        \includegraphics{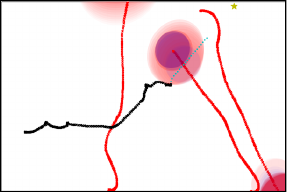}}
        \caption{$t=175$}
        \label{fig:hardware_175}
    \end{subfigure}
    \begin{subfigure}[t]{0.24\textwidth}
        \centering
        \scalebox{1.4}{
        \includegraphics{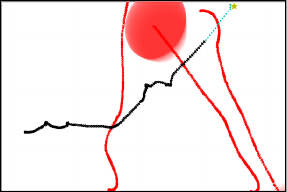}}
        \caption{$t=200$}
        \label{fig:hardware_200}
    \end{subfigure}
    \begin{subfigure}[t]{0.32\textwidth}
        \centering
        \scalebox{5}{
        \includegraphics{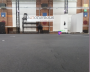}}
        \caption{Initial configuration, $t=1$}
        \label{fig:first_screenshot}
    \end{subfigure}
    \begin{subfigure}[t]{0.32\textwidth}
        \centering
        \scalebox{5}{
        \includegraphics{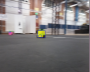}}
        \caption{Turning towards $\mathcal{O}_{1}$, $t=42$}
        \label{fig:middle_screenshot}
    \end{subfigure}
    \begin{subfigure}[t]{0.32\textwidth}
        \centering
        \scalebox{5}{
        \includegraphics{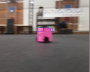}}
        \caption{Turning towards $\mathcal{O}_{2}$, $t=151$}
        \label{fig:last_screenshot}
    \end{subfigure}
    \caption{Trajectory followed by the robot when using velocity inputs output by the current iteration of the \algName{} motion planner. For both sets of figures, ``$t$" refers to the number of iterations of \algName{}.}
    \label{fig:plot_traj_hardware}
\end{figure*}

\begin{figure}
    \centering
    \scalebox{0.8}{
    \includegraphics{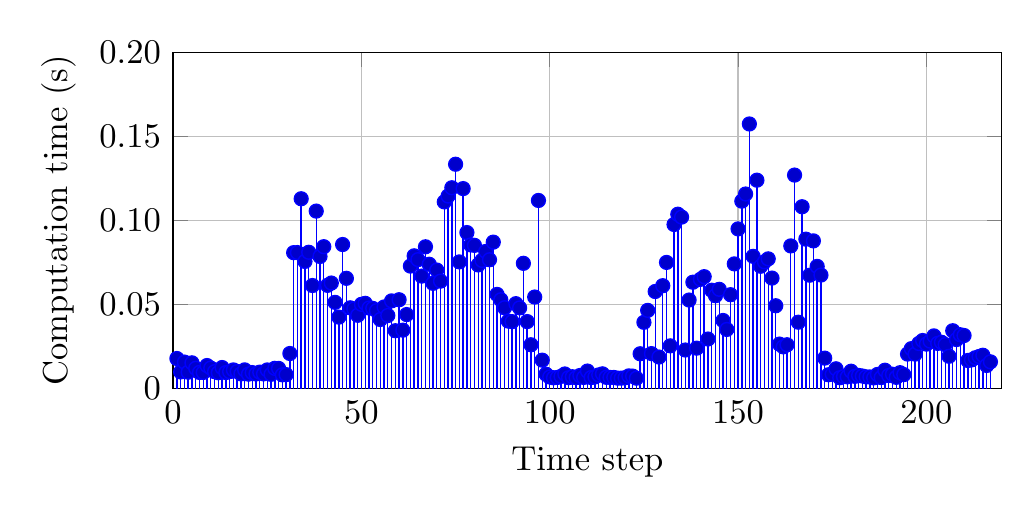}}
    \caption{Computation times for the \algName{} motion planner at each time step of the hardware experiment.}
    \label{fig:hardware_times}
\end{figure}

We now examine the feasibility of implementing the \algName{} motion planner in real-time hardware applications. Specifically, we consider the same robot and obstacle dynamics as in the previous example; however, rather than being given the next waypoint to visit, the robot is instead given the velocity and turning rate inputs $v[0]$ and $\omega[0]$ output from Algorithm 1. The robot then uses these inputs until the next iteration of \algName{} completes, and the process then repeats. Thus, each iteration of \algName{} must terminate sufficiently quickly such that the environment does not change too drastically before the completion of the subsequent iteration.  

The environment is displayed in Figure~\ref{fig:hardware_init}, where $x[0]=[-2.5,-0.75]^{\top}$ and $x_{g}=[1.9,1.9]^{\top}$. The set $\mathcal{S}$ is the $6$m$\times$$4$m rectangle centered at the origin. Once again, the obstacle parameters are provided in the Appendix. For the robot, we use a Turtlebot $3$ Waffle Pi connected to the Robot Operating System (ROS), setting $v[t]\in [0.01,0.25]$m/s and $\omega[t] \in [-\pi/2.5, \pi/2.5]$rad/s. Furthermore, the built-in Raspberry Pi camera allows a viewing angle of $\theta_{v}=1.09$rad. For this experiment, we set $dt=0.25$s use a planning horizon of $T=25$. All other parameters are identical to the previous example. Finally, we use a set of Turtlebot $3$ Burger Pis for the environment obstacles. To detect the obstacles at each sensing step in \algName, we map each one to a color and use OpenCV \cite{bradski2008learning} to search the acquired image for these colors. Specifically, we map $\mathcal{O}_{1}$ to yellow, $\mathcal{O}_{2}$ to green, and $\mathcal{O}_{3}$ to pink. If an obstacle is detected through this method, its precise location is provided using a motion capture system.

We run the experiment on a 2.6GHz Intel Core i7-9750H with 64 GB RAM\footnote{A video of the experiment is available at \url{https://youtu.be/GO8hFzfmb_0}.}. To reach the target state, the robot required a total of $217$ iterations of \algName{}. Figure~\ref{fig:plot_traj_hardware} plots the environment of the robot and its planned trajectory at several relevant time steps. At earlier time steps, the uncertainty in the position of obstacle $\mathcal{O}_{1}$ pushes the robot towards the top of $\mathcal{S}$. The robot then turns to observe $\mathcal{O}_{1}$ several times, as shown in Figures~\ref{fig:hardware_50} and~\ref{fig:hardware_75}, with the camera output shown in Figure~\ref{fig:middle_screenshot}. Afterwards, there is a brief period where the robot is able to travel unimpeded to the target state. Eventually, however, the uncertainty associated with obstacle $\mathcal{O}_{3}$ intersects the nominal trajectory of the robot. The robot then turns towards this obstacle, thereby reducing its uncertainty, as shown in Figure~\ref{fig:hardware_175} with camera output shown in Figure~\ref{fig:last_screenshot}. Afterwards, the robot is again able to travel unimpeded to the target state.
Figure~\ref{fig:hardware_times} displays the computation time per iteration of \algName{}. Note that all solve times are less than the $0.25$s time interval between discrete steps in~\eqref{eq:dubins_dyn}, indicating that \algName{} shows promise for online applications with collision-avoidance constraints.


\section{CONCLUSIONS}

We consider the problem of incorporating the relevance of an obstacle into the constrained sensor selection problem of a robot. 
To this end, we formulate the $\algName{}$ motion planner, wherein we use information about the dual variables of the linearized constraints corresponding to the ``keep-out" ellipsoids of each obstacle.
Through a sensitivity analysis, we show that measuring the state of the obstacle with the largest discounted sum of these dual variables should lead to the best gain in performance for the robot. We demonstrate the efficacy of the proposed motion planner in both software and hardware experiments, and find that \algName{} shows promise for future real-time hardware implementations.




\bibliographystyle{IEEEtran}
\bibliography{ref}

\appendix

\subsection{Obstacle details for the first simulation.}

For the first simulation experiment, the initial positions of the obstacles are given by
\begin{align*}
    & x_{1}[0]=\begin{bmatrix}
        3.00 \\ 
        0.25 \\
        0.25
    \end{bmatrix},
    & x_{2}[0]=\begin{bmatrix}
        -2.00 \\ 
        -2.00 \\
        -2.00
    \end{bmatrix},
    \phantom{-}
    & x_{3}[0]=\begin{bmatrix}
        -1.25 \\ 
        -1.25 \\
        -2.50
    \end{bmatrix}, \\
    & x_{4}[0]=\begin{bmatrix}
        3.00 \\ 
        1.75 \\
        1.75
    \end{bmatrix},
    & x_{5}[0]=\begin{bmatrix}
        -2.75 \\ 
        2.75 \\
        0.0
    \end{bmatrix}.
\end{align*}
Furthermore, the obstacles have Gaussian disturbances $w_{o}$ have mean vectors
\begin{align*}
    &\mu_{w_1} = \begin{bmatrix}
        -0.20 \\
        0.00 \\
        0.00
    \end{bmatrix},
    \mu_{w_2}=\begin{bmatrix}
        0.00 \\
        0.00 \\
        0.00 \\
    \end{bmatrix},
    \phantom{-}
    \mu_{w_3}=\begin{bmatrix}
        0.05 \\ 
        0.05 \\
        0.15
    \end{bmatrix}, \\
    &\mu_{w_4}=\begin{bmatrix}
        -0.03 \\
        0.00 \\
        0.00
    \end{bmatrix},
    \phantom{-}\mu_{w_5}=\begin{bmatrix}
        0.00 \\ 
        0.00 \\
        0.00
    \end{bmatrix},
\end{align*}
and covariance matrices
\begin{align*}
    \Sigma_{w_1} & = \begin{bmatrix} 0.01 & 0.001 & 0.001 \\ 0.001 & 0.01 & 0.001 \\ 0.001 & 0.001 & 0.01 \end{bmatrix}, 
    \Sigma_{w_2} = 0.01 \cdot \mathbf{I}_{3}, \\
    \Sigma_{w_3} & = \begin{bmatrix} 0.0125 & 0.0015 & 0.0015 \\ 0.0015 & 0.0125 & 0.0015 \\ 0.0015 & 0.0015 & 0.0125  \end{bmatrix},
    \Sigma_{w_4} = 0.06 \cdot \mathbf{I}_{3} \\
    \Sigma_{w_5} & = 0.01 \cdot \mathbf{I}_{3}.
\end{align*}

\subsection{Obstacle details for the second simulation.}

For the second simulation experiment, the initial positions of the obstacles are given by
\begin{align*}
    & x_{1}[0]=\begin{bmatrix}
        -2.00 \\ 
        2.00
    \end{bmatrix},
    & x_{2}[0]=\begin{bmatrix}
        -0.50 \\ 
        -2.00 
    \end{bmatrix},
    \phantom{-}
    & x_{3}[0]=\begin{bmatrix}
        2.75 \\ 
        -1.75 
    \end{bmatrix}.
\end{align*}
Furthermore, the obstacles have Gaussian disturbances $w_{o}$ have mean vectors
\begin{align*}
    &\mu_{w_1} = \begin{bmatrix}
        0.00 \\
        0.00
    \end{bmatrix},
    \phantom{-}
    \mu_{w_2}=\begin{bmatrix}
        0.05 \\
        0.15
    \end{bmatrix},
    \phantom{-}
    \mu_{w_3}=\begin{bmatrix}
        -0.1 \\ 
        0.0
    \end{bmatrix},
\end{align*}
and covariance matrices
\begin{align*}
    \Sigma_{w_1} & = \begin{bmatrix} 0.025 & 0.001 \\ 0.001 & 0.025 \end{bmatrix}, 
    \Sigma_{w_2} = \begin{bmatrix} 0.006 & 0.0015 \\ 0.0015 & 0.008 \end{bmatrix}, \\
    \Sigma_{w_3} & = \begin{bmatrix} 0.001 & 0.0 \\ 0.0 & 0.001 \end{bmatrix}.
\end{align*}

\subsection{Obstacle details for the hardware experiment.}

For the hardware experiment, the initial positions of the obstacles are given by
\begin{align*}
    & x_{1}[0]=\begin{bmatrix}
        -0.75 \\ 
        -2.00
    \end{bmatrix},
    & x_{2}[0]=\begin{bmatrix}
        1.20 \\ 
        1.80 
    \end{bmatrix},
    \phantom{-}
    & x_{3}[0]=\begin{bmatrix}
        2.25 \\ 
        -2.00 
    \end{bmatrix}.
\end{align*}
Furthermore, the obstacles have Gaussian disturbances $w_{o}$ have mean vectors
\begin{align*}
    &\mu_{w_1} = \begin{bmatrix}
        0.00 \\
        0.20
    \end{bmatrix},
    \phantom{-}
    \mu_{w_2}=\begin{bmatrix}
        0.05 \\
        -0.10
    \end{bmatrix},
    \phantom{-}
    \mu_{w_3}=\begin{bmatrix}
        -0.06 \\ 
        0.10
    \end{bmatrix},
\end{align*}
and covariance matrices
\begin{align*}
    \Sigma_{w_1} & = \begin{bmatrix} 0.008 & 0.001 \\ 0.001 & 0.008 \end{bmatrix}, 
    \Sigma_{w_2} = \begin{bmatrix} 0.004 & 0.0015 \\ 0.0015 & 0.005 \end{bmatrix}, \\
    \Sigma_{w_3} & = \begin{bmatrix} 0.008 & 0.0025 \\ 0.0025 & 0.0125 \end{bmatrix}.
\end{align*}

\end{document}